\documentclass[conference]{IEEEtran}
\IEEEoverridecommandlockouts
\usepackage{amsmath,amsfonts}
\usepackage{algorithmic}
\usepackage{algorithm}
\usepackage{array}
\usepackage[caption=false,font=normalsize,labelfont=sf,textfont=sf]{subfig}
\usepackage{textcomp}
\usepackage{stfloats}
\usepackage{url}
\usepackage{verbatim}
\usepackage{graphicx}
\usepackage{cite}
\usepackage{enumitem}
\usepackage[subnum]{cases}
\usepackage{multirow}
\usepackage{booktabs}
\usepackage{bm}
\usepackage{accents}
\usepackage{color}
\usepackage{amsmath}
\usepackage{amsthm}
\usepackage{makecell}
\usepackage[hypertexnames=false]{hyperref}
\makeatletter
\def \wideubar{\underaccent{{\cc@style\underline{\mskip15mu}}}}
\def \widebar{\accentset{{\cc@style\underline{\mskip10mu}}}}
\makeatother

\definecolor{blue}{rgb}{0,0,0.9}
\definecolor{red}{rgb}{0.9,0,0}
\definecolor{green}{rgb}{0,0.9,0}
\definecolor{lightgreen}{rgb}{0.1,0.5,0.1}

\begin{document}
\newtheorem{property}{Property}[section]
\newtheorem{proposition}{Proposition}[section]
\newtheorem{append}{Appendix}[section]
\newtheorem{definition}{Definition}[section]
\newtheorem{lemma}{Lemma}[section]
\newtheorem{corollary}{Corollary}[section]
\newtheorem{theorem}{Theorem}[section]
\newtheorem{remark}{Remark}[section]
\newtheorem{problem}{Problem}[section]
\newtheorem{example}{Example}[section]
\newtheorem{assumption}{Assumption}
\renewcommand*{\theassumption}{\Alph{assumption}}

\title{A Distributed Accelerated Preconditioned ADMM for Solving Quadratically Regularized Optimal Transport on Bipartite Graphs}

\author{\IEEEauthorblockN{Jiayi Zhu,~Haihang Lan$^{*}$}
\IEEEauthorblockA{\textit{School of Computer Science and Engineering} \\
\textit{Sun Yat-sen University}\\
Guangzhou 510006, China \\
e-mail: zhujy86@mail2.sysu.edu.cn, ~ lanhh3@mail2.sysu.edu.cn}
\thanks{$^*$ Corresponding author.}}

\maketitle

\begin{abstract}
Optimal transport (OT) has found broad applications in machine learning, economics, and management. Classical OT models, however, are not directly suited to allocation problems with graph-structured connectivity and unbalanced mass. Distributed OT formulations on bipartite graphs address these requirements but are typically solved by ADMM, whose slow convergence limits communication efficiency. We adapt the accelerated preconditioned ADMM (AP-ADMM) to the quadratically regularized distributed OT model and develop DAP-ADMM, a fully distributed algorithm in which each node communicates only with its graph neighbors. For fixed local penalty parameters, we establish global convergence and non-ergodic $O(1/k)$ bounds for both the KKT-residual norm and the primal objective gap. We derive an exact solver-free routine for the interval-constrained projections arising in the local ADMM subproblems. We propose a distributed self-adaptive penalty mechanism that updates node-specific penalties using locally computable KKT-residual components without global aggregation. Numerical experiments confirm the efficiency of the solver-free routine and the practical acceleration of DAP-ADMM over standard distributed ADMM.
\end{abstract}


\section{Introduction}

 {Optimal transport (OT), dating back to the seminal work of Monge \cite{monge1781memoire} and Kantorovich \cite{kantorovich1942translocation}, provides a fundamental framework for constructing cost-efficient couplings between probability distributions. Its ability to incorporate the geometry of the underlying sample space has led to a broad range of applications in machine learning \cite{courty2016optimal,kim2022optimal,luo2022differentiable,wang2022optimal}, computer vision \cite{rabin2015convex,solomon2015convolutional,choi2023generative}, image and signal processing \cite{sun2023spectr,cazelles2020wasserstein}, and statistics and data science \cite{zhang2024wasserstein,delon2020wasserstein,aburidi2024optimal}.}

Beyond these data-centric applications, OT also provides a natural allocation model by minimizing transportation costs between suppliers and demanders \cite{villani2008optimal,galichon2016optimal}. It therefore plays an important role in operations research, including logistics \cite{fajgelbaum2020optimal}, resource allocation \cite{huang2024convergence}, and economic planning and management. Applying OT to practical transportation networks, nevertheless, requires several departures from the classical model.

 {In many practical transportation networks, only a subset of supplier--demander pairs may be admissible because of physical reachability, service regions, or communication constraints \cite{essid2018quadratically,cang2022supervised,mokhtari2025admm}. The transportation plan must therefore respect a prescribed, often sparse bipartite graph.}

 {Moreover, the total supply and demand may not match exactly, and some applications require transporting only part of the available mass. These settings motivate unbalanced \cite{chizat2018scaling} and partial OT \cite{chapel2020partial}, or more generally the use of lower and upper bounds on the amount transported by each node.}

 {These modeling requirements have been incorporated into OT on bipartite graphs. In particular, the formulation in \cite{zhang2019consensus} represents each supply or demand entity as a node, uses admissible transportation links as graph edges, and replaces the equality constraints of classical OT with lower and upper transportation bounds. The same work further develops a distributed algorithm based on the alternating direction method of multipliers (ADMM) \cite{boyd2011distributed}, in which each node performs local computations and exchanges information only with its graph neighbors. This framework has subsequently been extended to several distributed OT settings \cite{hughes2022differentially,hughes2021fair,hughes2021resilient}. Despite this progress, three algorithmic challenges remain central to an efficient implementation:}

\vspace{1mm}

 {\noindent\textbf{Communication efficiency:} For conventional ADMM under general convexity, a standard non-ergodic complexity result bounds the best-iterate squared KKT residual by $O(1/k)$, corresponding to an $O(1/\sqrt{k})$ rate for the KKT residual norm \cite{cui2016convergence}. Consequently, obtaining a high-accuracy solution may still require many iterations. Because every distributed iteration entails information exchange between neighboring nodes, this iteration complexity directly translates into communication overhead.}

\vspace{1mm}

 {\noindent\textbf{Efficient solution of local ADMM subproblems:} Each node must solve a constrained local subproblem at every ADMM iteration. Existing distributed OT methods \cite{zhang2019consensus,hughes2021fair,hughes2021resilient,hughes2022differentially} generally rely on generic optimization routines for these repeated computations. For the quadratically regularized model considered here, completing the square reduces each local subproblem to a Euclidean projection onto the interval-constrained nonnegative simplex $\{\bm{x}\geq0\mid l\leq\mathbf{1}^{\top}\bm{x}\leq u\}$. Essid and Solomon \cite{essid2018quadratically} mentioned a related projection idea for the bipartite setting, but did not develop a complete algorithm for this two-sided interval constraint. Classical simplex and $\ell_{1}$-ball projection methods address the fixed-mass case recovered when the lower and upper bounds coincide \cite{held1974validation,condat2016fast,duchi2008efficient}. Consequently, a complete and efficient routine for the two-sided interval projection arising in every local update is still needed.} 

\vspace{1mm}

 {\noindent\textbf{Distributed tuning of local penalty parameters:} The numerical performance of ADMM-type methods is sensitive to the penalty parameters. Standard residual-based adaptation strategies typically update a common penalty parameter by comparing global primal--dual or KKT residuals \cite{boyd2011distributed,lam2018fast,yang2021fast}. Directly applying such strategies in a decentralized network would require collecting or aggregating residual information from all nodes. It is therefore nontrivial to adapt multiple node-specific penalty parameters while retaining a fully distributed implementation.}

 {To address these three challenges, this paper develops a distributed accelerated preconditioned ADMM framework for quadratically regularized OT on bipartite graphs, building on AP-ADMM \cite{syzz2025accelerating}. The main contributions are summarized as follows:}
\begin{enumerate}
	\item  {We derive a distributed implementation of AP-ADMM, termed DAP-ADMM, in which each supply or demand node performs local updates and communicates only with its graph neighbors. For fixed local penalty parameters, we establish global convergence and non-ergodic $O(1/k)$ bounds for the KKT residual and the primal objective gap;}
	\item  {We develop an exact \textit{solver-free} method for the local ADMM subproblems. The method explicitly handles distinct lower and upper bounds and has $O(p\log p)$ arithmetic complexity for a node of degree $p$;}
	\item  {We develop a distributed self-adaptation mechanism by decomposing the global KKT residual into locally computable residual components. Each node uses its local component to adjust its own penalty parameter, thereby avoiding global residual aggregation and preserving the fully distributed implementation of DAP-ADMM.}
\end{enumerate}

The remainder is organized as follows. Section \ref{sec:pre} formulates the quadratically regularized OT problem and its distributed reformulation. Section \ref{sec:DAP-ADMM} develops DAP-ADMM, including its distributed implementation, solver-free subsolver, convergence analysis, and local penalty adaptation. Section \ref{sec:num-exp} reports the numerical results, followed by the conclusion in Section \ref{sec:conclude}.

\vspace{1mm}
\noindent\textbf{Notation.}
{We use $\mathbb{R}^n$, $\mathbb{R}_{+}^n$, and $\mathbb{R}_{++}^n$ for real, nonnegative, and positive vectors, respectively, and $\mathbb{R}^{m\times n}$ for real matrices. The symbols $\bm{0}$, $\bm{1}$, $I$, and $O$ denote zero and all-one vectors, identity matrices, and zero matrices of appropriate dimensions; subscripts are omitted when the dimensions are clear. For vectors and matrices, $x_i$ and $X_{ij}$ denote entries, and $\|\bm{x}\|$ denotes the Euclidean norm.}

For a proper, closed, convex function $f$, $\operatorname{dom}f$, $\partial f$, and $f^{*}$ denote its effective domain, subdifferential, and convex conjugate, respectively. For a closed convex set $\mathcal{S}$, $\delta_{\mathcal{S}}$ and $\Pi_{\mathcal{S}}$ denote its indicator function and Euclidean projection. For $\nu>0$, we use $\mathtt{prox}_{\nu f}(\bm{x}) := \arg\min_{\bm{y}}\{f(\bm{y})+(2\nu)^{-1}\|\bm{y}-\bm{x}\|^{2}\}$.
In particular, $\mathtt{prox}_{\nu\delta_{\mathcal{S}}}(\bm{x})=\Pi_{\mathcal{S}}(\bm{x})$.


\section{Preliminaries}\label{sec:pre}

\subsection{Quadratically Regularized Optimal Transport on Bipartite Graphs}
Let $\mathcal{I}:=\{1,2,\cdots,m\}$ and $\mathcal{J}:=\{1,2,\cdots,n\}$ denote the supply and demand node index sets, respectively. Each supply node $i\in\mathcal{I}$ serves demand nodes in $\mathcal{J}_{i}$, while each demand node $j\in\mathcal{J}$ receives resources from supply nodes in $\mathcal{I}_{j}$. The resulting topology is the undirected connected bipartite graph $\mathcal{G}=\left\{\mathcal{I},\mathcal{J},\mathcal{E}\right\}$, where
\begin{equation*}
	\begin{aligned}
		\mathcal{E} &:= \left\{(i,j) ~ \big| \, \textrm{$i\in\mathcal{I}$ and $j\in\mathcal{J}_{i}$}\right\} \\
		&= \left\{(i,j) ~ \big| \, \textrm{$j\in\mathcal{J}$ and $i\in\mathcal{I}_{j}$}\right\},
	\end{aligned}
\end{equation*}
 {is the edge set. Thus, $\mathcal{I}_{j}$ and $\mathcal{J}_{i}$ are the supply and demand neighbor sets of $j\in\mathcal{J}$ and $i\in\mathcal{I}$, respectively.}

 {Each supply node $i\in\mathcal{I}$ can provide a finite quantity of resources with lower and upper bounds, denoted by $l_{i}^{s}$ and $u_{i}^{s}$ such that $0 \leq l_{i}^{s} \leq u_{i}^{s}$, respectively. Similarly, the resource requirement of each demand node $j\in\mathcal{J}$ also owns a pair of lower and upper bounds, denoted by $l_{j}^{d}$ and $u_{j}^{d}$, where $0 \leq l_{j}^{d} \leq u_{j}^{d}$. Our objective is to determine the transportation flow $\pi_{ij}$ on each edge $(i,j)\in\mathcal{E}$, ensuring that the supply and demand constraints are satisfied while the total transportation cost is minimized. For simplicity, let $\{x_{ij} \mid (i,j)\in\mathcal{E}\}$ be edge-associated elements and define $\bm{x}=[x_{ij}]_{(i,j)\in\mathcal{E}}$ with entries in lexicographic order.\footnote{Entries are ordered lexicographically: for each $i=1,\ldots,m$, sort $j\in\mathcal{J}_{i}$ increasingly, then proceed to $i+1$.}} We then denote by $\bm{x}_{(i,:)} = [x_{ij}]_{j\in\mathcal{J}_{i}}$ and $\bm{x}_{(:,j)} = [x_{ij}]_{i\in\mathcal{I}_{j}}$ the sub-vectors corresponding to nodes $i\in\mathcal{I}$ and $j\in\mathcal{J}$, respectively.

 {The problem is formulated as a quadratically regularized optimal transport problem on $\mathcal{G}$:}
\begin{equation}\label{eq:GOT}
	\begin{aligned}
		\min_{\{\pi_{ij}\}_{(i,j)\in\mathcal{E}}} & \quad \sum\limits_{(i,j)\in\mathcal{E}}c_{ij}\pi_{ij}+\frac{\lambda}{2}\pi_{ij}^{2} \\
		\mathrm{s.t.} \quad & \quad l_{i}^{s} \leq \sum\limits_{j \in \mathcal{J}_{i}} \pi_{ij} \leq u_{i}^{s},\quad \forall i\in\mathcal{I}, \\
		& \quad l_{j}^{d} \leq \sum\limits_{i \in \mathcal{I}_{j}} \pi_{ij} \leq u_{j}^{d},\quad \forall j\in\mathcal{J}, \\
		& \quad \pi_{ij} \geq 0,
	\end{aligned}
\end{equation}
where $c_{ij}$ is the unit transportation cost on edge $(i,j)\in\mathcal{E}$ and $\lambda\geq0$ is the quadratic-regularization parameter. If $\lambda=0$, $\mathcal{E}=\mathcal{I}\times\mathcal{J}$, and the capacity bounds coincide with prescribed marginals, i.e., $l_i^s=u_i^s=a_i$, $l_j^d=u_j^d=b_j$, and $\sum_i a_i=\sum_j b_j$, then \eqref{eq:GOT} reduces to the classical discrete OT problem. For $\lambda>0$, the quadratic term enhances the objective's strong convexity, preserves solution sparsity, and models nonlinear transportation costs \cite{essid2018quadratically,lorenz2021quadratically,zhang2019consensus}. We denote $\bm{\pi}=[\pi_{ij}]_{(i,j)\in\mathcal{E}}$ and $\bm{c}=[c_{ij}]_{(i,j)\in\mathcal{E}}$ as the transportation plan and cost vector, respectively.

\subsection{Distributed Reformulation of QROTs on Bipartite Graphs}
 {Following Zhang et al.~\cite{zhang2019consensus}, we reformulate \eqref{eq:GOT} for distributed ADMM:}
\begin{equation}\label{eq:D-GOT}
	\begin{aligned}
		\min_{\bm{s},\bm{d},\bm{\pi}\in\mathbb{R}^{|\mathcal{E}|}} & \quad \sum_{i\in\mathcal{I}}\sum_{j \in \mathcal{J}_{i}}\bigg(c_{ij}^{s}s_{ij} + \frac{\lambda^{s}}{2}s_{ij}^{2}\bigg) + \delta_{\mathcal{U}}(\bm{s}) \\
		& \quad  \; + \sum_{j\in\mathcal{J}}\sum_{i \in \mathcal{I}_{j}}\bigg(c_{ij}^{d}d_{ij} + \frac{\lambda^{d}}{2}d_{ij}^{2}\bigg) + \delta_{\mathcal{V}}(\bm{d}) \\
		\mathrm{s.t.} \quad & \quad \bm{s} = \bm{\pi}, ~~ \bm{\pi} = \bm{d}.
	\end{aligned}
\end{equation}
 {The reformulation introduces local copies $\bm{s},\bm{d}\in\mathbb{R}^{|\mathcal{E}|}$ of $\bm{\pi}$ for supply and demand nodes and consensus constraints $\bm{s}=\bm{\pi}$ and $\bm{\pi}=\bm{d}$. The objective and capacity constraints in \eqref{eq:GOT} then decompose as follows:}
\begin{itemize}
		\item  {The capacity constraints become indicator functions $\delta_{\mathcal{U}}(\cdot)$ and $\delta_{\mathcal{V}}(\cdot)$ of the convex sets $\mathcal{U}$ and $\mathcal{V}$:}
	\begin{equation*}
		\begin{aligned}[b]
			\mathcal{U} &:=\Big\{\bm{\pi}\in\mathbb{R}_{+}^{|\mathcal{E}|} ~ \big| ~ l_{i}^{s} \leq \sum\nolimits_{j \in \mathcal{J}_{i}}\pi_{ij} \leq u_{i}^{s}, ~~\forall i \in \mathcal{I}\Big\}, \\
			\mathcal{V} &:=\Big\{\bm{\pi}\in\mathbb{R}_{+}^{|\mathcal{E}|} ~ \big| ~ l_{j}^{d} \leq \sum\nolimits_{i \in \mathcal{I}_{j}} \pi_{ij} \leq u_{j}^{d}, ~~ \forall j\in\mathcal{J}\Big\} .
		\end{aligned}
	\end{equation*}
	
		\item  {The objective in \eqref{eq:GOT} splits into two parts:}
	\begin{align}
		G_{1}(\bm{\pi}) &:= \sum_{i\in\mathcal{I}}\sum_{j \in \mathcal{J}_{i}}\bigg(c_{ij}^{s}\pi_{ij} + \frac{\lambda^{s}}{2}\pi_{ij}^{2}\bigg) + \delta_{\mathcal{U}}(\bm{\pi}), \label{eq:lin-qua1}\\
		G_{2}(\bm{\pi}) &:= \sum_{j\in\mathcal{J}}\sum_{i \in \mathcal{I}_{j}}\bigg(c_{ij}^{d}\pi_{ij} + \frac{\lambda^{d}}{2}\pi_{ij}^{2}\bigg) + \delta_{\mathcal{V}}(\bm{\pi}), \label{eq:lin-qua2}
	\end{align}
	such that $c_{ij}^{s} + c_{ij}^{d} = c_{ij}$, $\lambda^{s} + \lambda^{d} = \lambda$ and $\lambda^{s},\lambda^{d} \geq 0$ for all $(i,j) \in\mathcal{E}$.
\end{itemize}
 {Thus, the ADMM updates are local and require only neighbor-to-neighbor communication; see \cite[Section 3]{zhang2019consensus}.}

\section{A Distributed Accelerated Preconditioned ADMM for QROT on Bipartite Graphs}\label{sec:DAP-ADMM}

In this section, we develop a distributed algorithm for solving problem \eqref{eq:D-GOT} based on the accelerated preconditioned ADMM (denoted by AP-ADMM), which is introduced in \cite[Algorithm 3.1]{syzz2025accelerating}. First, we introduce two stacks of local penalty parameters: $\bm{\sigma}^{s} = [\sigma_{1}^{s},\cdots,\sigma_{m}^{s}]$ for each supply node $i\in\mathcal{I}$, and $\bm{\sigma}^{d} = [\sigma_{1}^{d},\cdots,\sigma_{n}^{d}]$ for each demand node $j\in\mathcal{J}$, respectively. Then, denoting $\bm{\sigma} = [\bm{\sigma}^{s};\bm{\sigma}^{d}]$, the associated augmented Lagrangian function of problem \eqref{eq:D-GOT} with local penalty parameters is given by, for any $(\bm{\pi},\bm{s},\bm{d},\bm{\theta},\bm{\xi})\in\mathbb{R}^{|\mathcal{E}|}\times\mathbb{R}^{|\mathcal{E}|}\times\mathbb{R}^{|\mathcal{E}|}\times\mathbb{R}^{|\mathcal{E}|}\times\mathbb{R}^{|\mathcal{E}|}$:
\begin{equation}
	\begin{aligned}
		\mathcal{L}_{\bm{\sigma}}\big(\bm{\pi}, \bm{s}, \bm{d};\bm{\theta},\bm{\xi}\big) = & \; G_{1}(\bm{s}) + G_{2}(\bm{d}) \\
		& + \langle\bm{\theta},\bm{s}-\bm{\pi}\rangle + \langle\bm{\xi},\bm{\pi}-\bm{d}\rangle \\
		& + \frac{1}{2}\big\|\bm{s}-\bm{\pi}\big\|_{\Sigma_{1}}^{2} + \frac{1}{2}\big\|\bm{\pi}-\bm{d}\big\|_{\Sigma_{2}}^{2}, 
	\end{aligned}
\end{equation}
where 
\begin{equation}\label{eq:def_Sigma}
	\Sigma_{1} = \operatorname{Diag}(A_{1}^{\top}\bm{\sigma}^{s}), \quad \Sigma_{2} = \operatorname{Diag}(A_{2}^{\top}\bm{\sigma}^{d}),
\end{equation}
and $A_{1}\in\mathbb{R}^{m\times|\mathcal{E}|}$, $A_{2}\in\mathbb{R}^{n\times|\mathcal{E}|}$ are the incidence matrices associated with the bipartite graph $\mathcal{G}$ given by
\begin{equation*}
	\begin{aligned}
		\left[A_1\right]_{(i,(i,j))} &=\begin{cases}
			1, & \text { if } (i, j)\in\mathcal{E} \text { and } i \in \mathcal{I}, \\
			0, & \text { otherwise, }
		\end{cases} \\
		\left[A_2\right]_{(j,(i,j))} &= \begin{cases}
			1, & \text { if } (i, j)\in\mathcal{E} \text { and } j \in \mathcal{J}, \\
			0, & \text { otherwise. }\end{cases}
	\end{aligned}
\end{equation*}
Then, the AP-ADMM  for solving problem \eqref{eq:D-GOT} is now readily presented in Algorithm \ref{alg:acc-preADMM}.
\begin{algorithm}[h]
	\caption{AP-ADMM for solving \eqref{eq:D-GOT}}\label{alg:acc-preADMM}
	\begin{algorithmic}
		\STATE \textbf{Input:} Initialize $\bm{w}^{0}=(\bm{\pi}^{0}, \bm{s}^{0}, \bm{d}^{0}, \bm{\theta}^{0},\bm{\xi}^{0}) \in\mathbb{R}^{|\mathcal{E}|}\times\mathrm{dom}{G_1}\times\mathrm{dom}{G_2}\times\mathbb{R}^{|\mathcal{E}|}\times\mathbb{R}^{|\mathcal{E}|}$ and let $\widehat{\bm{w}}^{0}:=\bm{w}^{0}$. Set parameters $\bm{\sigma} = [\bm{\sigma}^{s}; \bm{\sigma}^{d}] \in\mathbb{R}_{++}^{m+n}$ and $\rho\in(0,2]$. Set $k = 0$.
		\WHILE{a termination criterion is not met,}
		\STATE \textbf{Step 1.} Compute
		\begin{equation}\label{eq:DAP-ADMM-SD}
			\big(\widebar{\bm{s}}^{k},\widebar{\bm{d}}^{k}\big) 
			= \mathop{\arg\min}\limits_{\bm{s},\bm{d}\in\mathbb{R}^{|\mathcal{E}|}} ~
			\mathcal{L}_{\bm{\sigma}}\big(\bm{\pi}^{k},\bm{s},\bm{d};\bm{\theta}^{k}, \bm{\xi}^{k}\big).
		\end{equation}
		
		\STATE \textbf{Step 2.} Set 
		\begin{align*}
			\widebar{\bm{\theta}}^{k} &= \bm{\theta}^{k} + \Sigma_{1}\big(\widebar{\bm{s}}^{k}-\bm{\pi}^{k}\big), \\
			\widebar{\bm{\xi}}^{k} &= \bm{\xi}^{k}+\Sigma_{2}\big(\bm{\pi}^{k}-\widebar{\bm{d}}^{k}\big).
		\end{align*}
		
		\STATE \textbf{Step 3.} Compute
		\begin{equation}\label{eq:DAP-ADMM-X}
			\widebar{\bm{\pi}}^{k} = \mathop{\arg\min}\limits_{\bm{\pi}\in\mathbb{R}^{|\mathcal{E}|}} ~
			\mathcal{L}_{\bm{\sigma}}\big(\bm{\pi}, \widebar{\bm{s}}^{k}, \widebar{\bm{d}}^{k};
			\widebar{\bm{\theta}}^{k}, \widebar{\bm{\xi}}^{k}\big).
		\end{equation}
		
		\STATE \textbf{Step 4.} Set $\widehat{\bm{w}}^{k+1} = \bm{w}^{k} + \rho\big(\widebar{\bm{w}}^{k}-\bm{w}^{k}\big)$. \vspace{1.2mm}
		
		\STATE \textbf{Step 5.} Set $\bm{w}^{k+1} = \frac{1}{k+2}\bm{w}^{0} + \frac{k+1}{k+2}\widehat{\bm{w}}^{k+1}$. \vspace{1.2mm}
		
		\STATE \textbf{Step 6.} Set $k \leftarrow k+1$ and go to \textbf{Step 1}.
		\ENDWHILE
		\STATE \textbf{Output:} $\widebar{\bm{w}}^{k}=\big(\widebar{\bm{\pi}}^{k}, \widebar{\bm{s}}^{k}, \widebar{\bm{d}}^{k};
		\widebar{\bm{\theta}}^{k}, \widebar{\bm{\xi}}^{k}\big)$.
	\end{algorithmic}
\end{algorithm}

\subsection{DAP-ADMM: Distributed Implementations of AP-ADMM}\label{sec:implementations}

All operations in Algorithm \ref{alg:acc-preADMM} are distributed: each supply node $i\in\mathcal{I}$ and demand node $j\in\mathcal{J}$ exchanges information only with its neighbors, and the computations are performed locally.  We now describe the distributed implementation of each computational step in Algorithm \ref{alg:acc-preADMM} in detail.
\begin{itemize}[labelsep=-1em]
\item[\textbf{Step}] \;\quad \textbf{1.} The updates of $\big(\widebar{\bm{s}}^{k},\widebar{\bm{d}}^{k}\big)$ is obtained by solving ADMM-subproblem \eqref{eq:DAP-ADMM-SD}, which is formulated as
\begin{equation*}
\begin{aligned}
& \min\limits_{\bm{s},\bm{d}\in\mathcal{R}^{|\mathcal{E}|}} ~
\mathcal{L}_{\bm{\sigma}}\big(\bm{\pi}^{k},\bm{s},\bm{d};\bm{\theta}^{k}, \bm{\xi}^{k}\big) \\
= & \min\limits_{\bm{s},\bm{d}\in\mathcal{R}^{|\mathcal{E}|}} ~ \left\{\begin{array}{l}\displaystyle
G_{1}(\bm{s}) + G_{2}(\bm{d})	\\
\; + \frac{1}{2}\big\|\bm{s}-\bm{\pi}^{k}+\Sigma_{1}^{-1}\bm{\theta}^{k}\big\|_{\Sigma_{1}}^{2} \\
\; + \frac{1}{2}\big\|\bm{\pi}^{k}-\bm{d}+\Sigma_{2}^{-1}\bm{\xi}^{k}\big\|_{\Sigma_{2}}^{2}
\end{array}\right\},
\end{aligned}
\end{equation*}
and the joint update is separable over $\{\bm{s}_{(i,:)}\}_{i\in\mathcal{I}}$ and $\{\bm{d}_{(:,j)}\}_{j\in\mathcal{J}}$. Specifically, supply node $i\in\mathcal{I}$ computes
\begin{equation}\label{eq:DAP-ADMM-S}
\begin{aligned}
& \widebar{\bm{s}}_{(i,:)}^{k} = \mathop{\arg\min}\limits_{\bm{s}_{(i,:)}\in\mathcal{U}_{i}} \left\{\begin{array}{l}
\sum_{j \in \mathcal{J}_{i}}c_{ij}^{s}s_{ij} + \frac{\lambda^{s}}{2}s_{ij}^{2} \\
+ \frac{\sigma_{i}^{s}}{2}\big\|\bm{s}_{(i,:)} - \bm{\pi}_{(i,:)}^{k} + \frac{\bm{\theta}_{(i,:)}^{k}}{\sigma_{i}^{s}}\big\|^{2}
\end{array}\right\},
\end{aligned}
\end{equation}
where $\sigma_{i}^{s}$ is the local penalty parameter, and
\begin{equation*}
\textstyle\mathcal{U}_{i}:=\left\{\bm{s}_{(i,:)} \mid s_{i j} \geq 0, l_i^s \leq\sum_{j \in \mathcal{J}_{i}} s_{i j} \leq u_i^s\right\},
\end{equation*}
is the slice of $\mathcal{U}$ associated with supply node $i\in\mathcal{I}$. Demand node $j\in\mathcal{J}$ analogously computes $\widebar{\bm{d}}_{(:,j)}^{k}$ by
\begin{equation}\label{eq:DAP-ADMM-D}
\begin{aligned}
& \widebar{\bm{d}}_{(:,j)}^{k} = \mathop{\arg\min}\limits_{\bm{d}_{(:,j)}\in\mathcal{V}_{j}} \left\{\begin{array}{l}
\sum_{i \in \mathcal{I}_{j}}c_{ij}^{d}d_{ij} + \frac{\lambda^{d}}{2}d_{ij}^{2} \\
+ \frac{\sigma_{j}^{d}}{2}\big\|\bm{\pi}_{(:,j)}^{k} - \bm{d}_{(:,j)} + \frac{\bm{\xi}_{(:,j)}^{k}}{\sigma_{j}^{d}}\big\|^{2}
\end{array}\right\},
\end{aligned}
\end{equation}
where $\sigma_{j}^{d}$ is the local penalty parameter, and 
\begin{equation*}
\textstyle\mathcal{V}_{j}:=\left\{\bm{d}_{(:,j)} \mid d_{i j} \geq 0, l_j^d \leq\sum_{i \in \mathcal{I}_{j}}  {d_{i j}} \leq u_j^d\right\},
\end{equation*}
is the slice of $\mathcal{V}$ associated with demand node $j\in\mathcal{J}$. These local linear-quadratic subproblems are the main per-node computational cost. 
	
\item[\textbf{Step}] \;\quad \textbf{2.} The updates of Lagrangian multipliers $\widebar{\bm{\theta}}^{k}$ and $\widebar{\bm{\xi}}^{k}$ are naturally separable: supply node $i\in\mathcal{I}$ updates $\widebar{\bm{\theta}}_{(i,:)}^{k}$ locally
\begin{equation}\label{eq:DAP-ADMM-Theta}
    \widebar{\bm{\theta}}_{(i,:)}^{k} = \bm{\theta}_{(i,:)}^{k}+\sigma_{i}^{s}\big(\widebar{\bm{s}}_{(i,:)}^{k}-\bm{\pi}_{(i,:)}^{k}\big), 
\end{equation}
Demand node $j\in\mathcal{J}$ updates $\widebar{\bm{\xi}}_{(:,j)}^{k}$ locally
\begin{equation}\label{eq:DAP-ADMM-Xi}
    \widebar{\bm{\xi}}_{(:,j)}^{k} = \bm{\xi}_{(:,j)}^{k}+\sigma_{j}^{d}\big(\bm{\pi}_{(:,j)}^{k}-\widebar{\bm{d}}_{(:,j)}^{k}\big). 
\end{equation}	
	
\item[\textbf{Step}] \;\quad \textbf{3.} The update of $\widebar{\bm{\pi}}^{k}$ is obtained by solving ADMM-subproblem \eqref{eq:DAP-ADMM-X}, which admits a closed-form solution
\begin{equation}\label{eq:DAP-ADMM-X2}
\widebar{\bm{\pi}}^{k} = (\Sigma_{1}+\Sigma_{2})^{-1}\big(\Sigma_{1}\widebar{\bm{s}}^{k} + \Sigma_{2}\widebar{\bm{d}}^{k} + \widebar{\bm{\theta}}^{k} - \widebar{\bm{\xi}}^{k}\big).
\end{equation}
Because $\Sigma_{1}$ and $\Sigma_{2}$ are diagonal, this update is computed edge-wise:
\begin{align*}
\bar{\pi}_{ij}^{k} = \frac{\sigma_{i}^{s}\bar{s}_{ij}^{k} + \sigma_{j}^{d}\bar{d}_{ij}^{k} + \bar{\theta}_{ij}^{k} - \bar{\xi}_{ij}^{k}}{\sigma_{i}^{s}+\sigma_{j}^{d}},
\end{align*}
Each supply node collects $(\sigma_{j}^{d},\,\bar{d}_{ij},\bar{\xi}_{ij})$ from its demand neighbors and computes its row sub-vector; each demand node performs the analogous column update. Both sides obtain the same $\bar{\pi}_{ij}^{k}$ that preserves consensus.
	
\item[\textbf{Step}] \;\quad \textbf{4}-\textbf{5.} Since the over-relaxation in \textbf{Step 4} and Halpern iteration in \textbf{Step 5} involve only element-wise operations, their computations are naturally separable, too. Specifically, the computations of \textbf{Step 4}-\textbf{5} can be decomposed as follows, for every $i\in\mathcal{I}$:
	\begin{equation}\label{eq:ex-Hal-S}
		\begin{aligned}
			\widehat{\bm{w}}_{(i,:)}^{k+1} & = \bm{w}_{(i,:)}^{k} + \rho\big(\widebar{\bm{w}}_{(i,:)}^{k} - \bm{w}_{(i,:)}^{k}\big), \\
			\bm{w}_{(i,:)}^{k+1} & = \frac{1}{k+2}\bm{w}_{(i,:)}^{0} + \frac{k+1}{k+2}\widehat{\bm{w}}_{(i,:)}^{k+1},
		\end{aligned}
	\end{equation}
	and, for every $j\in\mathcal{J}$,
	\begin{equation}\label{eq:ex-Hal-D}
		\begin{aligned}
			\widehat{\bm{w}}_{(:,j)}^{k+1} & = \bm{w}_{(:,j)}^{k} + \rho\big(\widebar{\bm{w}}_{(:,j)}^{k} - \bm{w}_{(:,j)}^{k}\big), \\
			\bm{w}_{(:,j)}^{k+1} & = \frac{1}{k+2}\bm{w}_{(:,j)}^{0} + \frac{k+1}{k+2}\widehat{\bm{w}}_{(:,j)}^{k+1}.
		\end{aligned}
	\end{equation}
\end{itemize}
Based on the above implementations, we propose the distributed accelerated preconditioned ADMM, denoted as DAP-ADMM, for solving problem \eqref{eq:D-GOT} in Algorithm \ref{alg:DAP-ADMM}.

\begin{algorithm}[h]
	\caption{\textbf{DAP-ADMM} for solving \eqref{eq:D-GOT}}\label{alg:DAP-ADMM}
	\begin{algorithmic}
		\STATE \textbf{Input:} Initialize $\bm{w}^{0}=(\bm{\pi}^{0}, \bm{s}^{0}, \bm{d}^{0}, \bm{\theta}^{0},\bm{\xi}^{0}) \in\mathbb{R}^{|\mathcal{E}|}\times\mathrm{dom}{G_1}\times\mathrm{dom}{G_2}\times\mathbb{R}^{|\mathcal{E}|}\times\mathbb{R}^{|\mathcal{E}|}$ and let $\widehat{\bm{w}}^{0}:=\bm{w}^{0}$. Set parameters $\bm{\sigma} = [\bm{\sigma}^{s}; \bm{\sigma}^{d}] \in\mathbb{R}_{++}^{m+n}$ and $\rho\in(0,2]$. 
		\FOR{$k = 0, 1, 2, \ldots $,}
		\STATE  {\textbf{Parallel node updates:} In each of the following steps, all supply nodes $i\in\mathcal{I}$ and all demand nodes $j\in\mathcal{J}$ operate in parallel.}
		\STATE \textbf{Step 1.} node $i$ updates $\widebar{\bm{s}}_{(i,:)}^{k}$ via solving \eqref{eq:DAP-ADMM-S};  \\
		\qquad \quad \;\; node $j$ updates $\widebar{\bm{d}}_{(:,j)}^{k}$ via solving \eqref{eq:DAP-ADMM-D}; 
		
		\STATE \textbf{Step 2.} node $i$ updates $\widebar{\bm{\theta}}_{(i,:)}^{k}$ via \eqref{eq:DAP-ADMM-Theta}; \\
		\qquad \quad \;\; node $j$ updates $\widebar{\bm{\xi}}_{(:,j)}^{k}$ via \eqref{eq:DAP-ADMM-Xi};
		
		\STATE \textbf{Communication.} \\
		node $i$ collects $(\sigma_{j}^{d},\, \bar{d}_{ij}^{k}, \bar{\xi}_{ij}^{k})$ from its neighbors $\mathcal{J}_{i}$; node $j$ collects $(\sigma_{i}^{s},\, \bar{s}_{ij}^{k}, \bar{\theta}_{ij}^{k})$ from its neighbors $\mathcal{I}_{j}$;
		
		\STATE \textbf{Step 3.} node $i$ updates $\widebar{\bm{\pi}}_{(i,:)}^{k}$  via \eqref{eq:DAP-ADMM-X2}; \\
		\qquad \quad \;\; node $j$ updates $\widebar{\bm{\pi}}_{(:,j)}^{k}$ via \eqref{eq:DAP-ADMM-X2};
		
		\STATE \textbf{Step 4.} node $i$ computes $\widehat{\bm{w}}_{(i,:)}^{k+1}$ and $\bm{w}_{(i,:)}^{k+1}$ via \eqref{eq:ex-Hal-S}; \\
		\qquad \quad \;\; node $j$ computes $\widehat{\bm{w}}_{(:,j)}^{k+1}$ and $\bm{w}_{(:,j)}^{k+1}$ via \eqref{eq:ex-Hal-D}.
		
		\ENDFOR
		\STATE \textbf{Output:} $\widebar{\bm{w}}^{k}=\big(\widebar{\bm{\pi}}^{k}, \widebar{\bm{s}}^{k}, \widebar{\bm{d}}^{k};
		\widebar{\bm{\theta}}^{k}, \widebar{\bm{\xi}}^{k}\big)$.
	\end{algorithmic}
\end{algorithm}



\subsection{A Solver-Free Method for Solving ADMM-Subproblems }\label{sec:App-subsolver-free}
In this section, we introduce an efficient \textit{solver-free} method for solving the ADMM-subproblems \eqref{eq:DAP-ADMM-S} and \eqref{eq:DAP-ADMM-D}. For cost functions $G_{1}$ and $G_{2}$ given in \eqref{eq:lin-qua1} and \eqref{eq:lin-qua2}, the ADMM-subproblems could be rewritten as:
\begin{equation}\label{eq:DAP-ADMM-Si}
	\begin{aligned}[b]
		& \widebar{\bm{s}}_{(i,:)}^{k} = \mathop{\arg\min}\limits_{\bm{s}_{(i,:)}\in\mathcal{U}_{i}} \\
		& \textstyle\left\{\frac{\lambda^{s} + \sigma_{i}^{s}}{2}\big\|\bm{s}_{(i,:)} -  \frac{\sigma_{i}^{s}\bm{\pi}_{(i,:)}^{k} - \bm{\theta}_{(i,:)}^{k} - \bm{c}_{(i,:)}^{s} }{\lambda^{s} + \sigma_{i}^{s}}\big\|^{2}\right\},
	\end{aligned}
\end{equation}
and
\begin{equation}\label{eq:DAP-ADMM-Dj}
	\begin{aligned}[b]
		& \widebar{\bm{d}}_{(:,j)}^{k} = \mathop{\arg\min}\limits_{\bm{d}_{(:,j)}\in\mathcal{V}_{j}} \\
		& \textstyle\left\{\frac{\lambda^{d} + \sigma_{j}^{d}}{2}\big\|\bm{d}_{(:,j)} -  \frac{\sigma_{j}^{d}\bm{\pi}_{(:,j)}^{k} + \bm{\xi}_{(:,j)}^{k} - \bm{c}_{(:,j)}^{d} }{\lambda^{d} + \sigma_{j}^{d}}\big\|^{2}\right\}.
	\end{aligned}
\end{equation}
Both of them fall into the form of an interval-constrained projection problem:
\begin{equation}\label{eq:simplex-like}
	\min\limits_{\bm{x}\in\mathbb{R}^{p}} \quad \frac{\gamma}{2}\big\|\bm{x} - \widehat{\bm{x}}\big\|^{2}, \quad \mathrm{s.t.} \quad l \leq \mathbf{1}^{\top}\bm{x} \leq u, \quad \bm{x}\geq0.
\end{equation}
Taking subproblem \eqref{eq:DAP-ADMM-Si} for an arbitrary $i\in\mathcal{I}$ as an example. Let $p = |\mathcal{J}_{i}|$ and denote the parameters in \eqref{eq:simplex-like} as
\begin{align*}
	& \gamma = \lambda^{s}+\sigma_{i}^{s}, \quad l = l_{i}^{s}, \quad u = u_{i}^{s}, \\
	& \widehat{\bm{x}} = \frac{\sigma_{i}^{s} \bm{\pi}_{(i,:)}^{k} - \bm{\theta}_{(i,:)}^{k} - \bm{c}_{(i,:)}^{s}}{\lambda^{s}+\sigma_{i}^{s}},
\end{align*}
then \eqref{eq:DAP-ADMM-Si} is exactly equivalent to \eqref{eq:simplex-like}.

We now focus on the solution to problem \eqref{eq:simplex-like}. According to \cite[Corollary 28.3.1]{r1970convex}, we know that $\bm{x}^{*}\in\mathbb{R}^{p}$ is the solution to \eqref{eq:simplex-like} if and only if there exists $(\mu_1^{*}, \mu_2^{*}, \bm{\xi}^{*})\in\mathbb{R}\times\mathbb{R}\times\mathbb{R}^{p}$ such that $(\bm{x}^{*},\mu_1^{*}, \mu_2^{*}, \bm{\xi}^{*})$ satisfies the following KKT condition:
\begin{subequations}
\begin{align}
& \gamma(\bm{x}^{*}-\widehat{\bm{x}})-\mu_{1}^{*}\mathbf{1}+\mu_{2}^{*}\mathbf{1}-\bm{\zeta}^{*} = 0, \label{kkt:sta}\\
& \mathbf{1}^{\top}\bm{x}^{*} - l \geq 0, \quad u - \mathbf{1}^{\top}\bm{x}^{*} \geq 0, \quad \bm{x}^{*} \geq 0, \label{kkt:pfeas}\\
& \mu_{1}^{*} \geq 0, \quad \mu_{2}^{*} \geq 0, \quad \bm{\zeta}^{*} \geq 0, \label{kkt:dfeas}\\
& \mu_{1}^{*}(\mathbf{1}^{\top}\bm{x}^{*} - l) = 0, \quad \mu_{2}^{*}(\mathbf{1}^{\top}\bm{x}^{*} - u) = 0, \label{kkt:cs1}\\
& \forall i = 1,2,\cdots,p, \quad x_{i}^{*}\zeta_{i}^{*} = 0. \label{kkt:cs2}
\end{align}
\end{subequations}
where $\mu_{1},\mu_{2}\in\mathbb{R}$ and $\bm{\zeta}\in\mathbb{R}^{p}$ are the Lagrangian multipliers with respect to the lower and upper inequality constraints and non-negative constraints, respectively. 

From the stationary condition \eqref{kkt:sta} and the complementary slackness \eqref{kkt:cs2}, we deduce the following relationships for each coordinate $i = 1,2,\cdots,p$:
\begin{itemize}
	\item if $x_{i}^{*}>0$, then $\zeta_{i}^{*} = 0$ and $x_{i}^{*} = \widehat{x}_{i} + \gamma^{-1}(\mu_{1}^{*}-\mu_{2}^{*}) > 0$;
	\item if $x_{i}^{*}=0$, then $\zeta_{i}^{*} \geq 0$ and $x_{i}^{*} = \widehat{x}_{i} + \gamma^{-1}(\mu_{1}^{*}-\mu_{2}^{*}) + \gamma^{-1}\zeta_{i}^{*} = 0$.
\end{itemize}
Thus, the relationship between $\bm{x}$ and $\mu$ can be concluded as:
\begin{equation*}\label{eq:x-mu}
	\bm{x}(\mu) = \max\{\widehat{\bm{x}} + \mu\mathbf{1}, 0\}, ~~ \text{where} ~~ \mu=\gamma^{-1}(\mu_{1}-\mu_{2}).
\end{equation*}
Then, the optimal $\bm{x}^{*}$ could be obtained by identify a $\mu^{*}\in\mathbb{R}$ such that $l \leq \mathbf{1}^{\top}\bm{x}(\mu^{*}) \leq u$. To find such a $\mu^{*}$, we first evaluate $\mathbf{1}^{\top}\bm{x}(0)$ to determine whether it lies in $[l,u]$:
\begin{itemize}
\item If $\mathbf{1}^{\top}\bm{x}(0)\in[l,u]$, then $\bm{x}^{*}=\max\{\widehat{\bm{x}},0\}$ and $\mu^{*}=\mu_{1}^{*}=\mu_{2}^{*}=0$.
\item If $\mathbf{1}^{\top}\bm{x}(0)<l$, the lower bound is active, with $\mu^{*}>0$, $\mu_{1}^{*}>0$, and $\mu_{2}^{*}=0$; then \eqref{eq:simplex-like} reduces to the equality-constrained simplex projection
\begin{equation}\label{eq:simplex-l}
    \min\limits_{\bm{x}\in\mathbb{R}^{p}} \quad \frac{\gamma}{2}\big\|\bm{x} - \widehat{\bm{x}}\big\|^{2}, \quad \mathrm{s.t.} \quad \mathbf{1}^{\top}\bm{x} = l, \quad \bm{x}\geq0,
\end{equation}
We solve this simplex projection with the sorting-based method of \cite{held1974validation}; related algorithms are reviewed in \cite{condat2016fast}. It computes $\mu^{*}$ by:
\begin{quote}
\begin{verbatim}
temp = cumsum(sort(x_hat, ...
1, 'descend'), 1);
mu = - max((temp - l) ./ (1:p)');
\end{verbatim}
\end{quote}
The resulting $\bm{x}^{*}$ is computed by $\bm{x}^{*}=\max\{\widehat{\bm{x}}+\mu^{*}\mathbf{1},0\}$, with $\mu_{1}^{*}=\gamma\mu^{*}$ and $\mu_{2}^{*}=0$. The sorting step dominates the computation and has $O(p\log p)$ arithmetic complexity.
\item If $\mathbf{1}^{\top}\bm{x}(0)>u$, the upper bound is active, with $\mu^{*}<0$, $\mu_{1}^{*}=0$, and $\mu_{2}^{*}>0$; then \eqref{eq:simplex-like} reduces to
\begin{equation}\label{eq:simplex-u}
    \min\limits_{\bm{x}\in\mathbb{R}^{p}} \quad \frac{\gamma}{2}\big\|\bm{x} - \widehat{\bm{x}}\big\|^{2}, \quad \mathrm{s.t.} \quad \mathbf{1}^{\top}\bm{x} = u, \quad \bm{x}\geq0.
\end{equation}
\end{itemize}
Finally, we summarize the entire procedure for solving problem \eqref{eq:simplex-like} in Algorithm \ref{alg:simplex-like}.

\begin{algorithm}[h]
\caption{Solver-free method for solving problem \eqref{eq:simplex-like}}\label{alg:simplex-like}
\begin{algorithmic}
\STATE \textbf{Input:} the modulus $\gamma > 0$, the vector $\widehat{\bm{x}}\in\mathbb{R}^{p}$, and the lower and upper bounds $l$ and $u$.
\STATE \quad \textbf{Step 1.} compute $\bm{x}(0)$ and check the value of $\mathbf{1}^{\top}\bm{x}(0)$;
\STATE \quad \textbf{Step 2a.} if $\mathbf{1}^{\top}\bm{x}(0) \in [l,u]$, then set $\bm{x}^{*} = \max\{\widehat{\bm{x}}, 0\}$ and set $\mu_{1}^{*}=\mu_{2}^{*}=0$;
\STATE \quad \textbf{Step 2b.} if $\mathbf{1}^{\top}\bm{x}(0) < l$, then compute $(\bm{x}^{*},\mu_{1}^{*})$ by solving \eqref{eq:simplex-l} and set $\mu_{2}^{*}=0$;
\STATE \quad \textbf{Step 2c.} if $\mathbf{1}^{\top}\bm{x}(0) > u$, then compute $(\bm{x}^{*},\mu_{2}^{*})$ by solving \eqref{eq:simplex-u} and set $\mu_{1}^{*}=0$;
\STATE \textbf{Output:} $(\bm{x}^{*}$, $\mu_{1}^{*}$, $\mu_{2}^{*})$.
\end{algorithmic}
\end{algorithm}

\subsection{Convergence of DAP-ADMM for Solving \eqref{eq:D-GOT}}\label{sec:thms}
In this section, we analyze the convergence properties of DAP-ADMM based on the theories developed in \cite{syzz2025accelerating}. To this end, we first identify the dual problem of \eqref{eq:D-GOT} as
\begin{equation}\label{eq:D-GOT-dual}
\begin{aligned}
\min_{\bm{\theta},\bm{\xi}\in\mathbb{R}^{|\mathcal{E}|}} \quad G_{1}^{*}(-\bm{\theta}) + G_{2}^{*}(\bm{\xi}) \quad \mathrm{ s.t. } \quad \bm{\theta} = \bm{\xi},
\end{aligned}
\end{equation}
where $G_{1}^{*}$ and $G_{2}^{*}$ are the convex conjugates of $G_{1}$ and $G_{2}$. For notational convenience, we denote
\begin{equation*}
	\begin{aligned}
		\bm{w} &:= \big(\bm{\pi},\bm{s},\bm{d},\bm{\theta},\bm{\xi}\big), \\
		\mathbb{W} &:= \mathbb{R}^{|\mathcal{E}|}\times\mathbb{R}^{|\mathcal{E}|}\times\mathbb{R}^{|\mathcal{E}|}\times\mathbb{R}^{|\mathcal{E}|}\times\mathbb{R}^{|\mathcal{E}|}.
	\end{aligned}
\end{equation*}
Then, it follows from \cite[Corollary 28.3.1]{r1970convex} that $(\bm{\pi}^{*}, \bm{s}^*, \bm{d}^*) \in \mathbb{R}^{|\mathcal{E}|}\times\mathbb{R}^{|\mathcal{E}|}\times\mathbb{R}^{|\mathcal{E}|}$ is optimal to problem \eqref{eq:D-GOT} if and only if there exists $(\bm{\theta}^*,\bm{\xi}^*) \in \mathbb{R}^{|\mathcal{E}|}\times\mathbb{R}^{|\mathcal{E}|}$ such that $\big(\bm{\pi}^{*},\bm{s}^*,\bm{d}^*,\bm{\theta}^*,\bm{\xi}^*\big)$ satisfies the following KKT system:
\begin{equation}\label{eq:KKTsys}
	\begin{aligned}[b]
		\bm{0} &= -\bm{\theta}^*+\bm{\xi}^*, \\
		\bm{0} &\in\partial G_{1}(\bm{s}^*) + \bm{\theta}^*, \\
		\bm{0} &\in\partial G_{2}(\bm{d}^*) - \bm{\xi}^*, \\
		\bm{0} &= \bm{\pi}^{*} - \bm{s}^*, \\
		\bm{0} &= \bm{d}^* - \bm{\pi}^{*},
	\end{aligned}
\end{equation}
where $\partial G_{1}$ and $\partial G_{2}$ are the subdifferential mappings of functions $G_{1}$ and $G_{2}$, respectively. Then, solving problem \eqref{eq:D-GOT} is equivalent to finding $\bm{w}\in\mathbb{W}$ such that $\bm{0}\in\mathcal{T}(\bm{w})$, where the maximal monotone operator $\mathcal{T}$ is defined, for every $\bm{w} = (\bm{\pi},\bm{s},\bm{d},\bm{\theta},\bm{\xi})\in\mathbb{W}$, by
\begin{equation}\label{eq:KKTmap-GOT}
	\mathcal{T}(\bm{w}) = \begin{pmatrix}
		-\bm{\theta}+\bm{\xi} \\
		\partial G_{1}(\bm{s}) + \bm{\theta} \\
		\partial G_{2}(\bm{d}) - \bm{\xi} \\
		\bm{\pi} - \bm{s} \\
		\bm{d} - \bm{\pi}
	\end{pmatrix}.
\end{equation}
Now, we make the following assumption:
\begin{assumption}\label{asp:KKT-nonempty}
	The solution set to the KKT system \eqref{eq:KKTsys} is nonempty.
\end{assumption}
This assumption is very mild and standard in the analysis of convex constrained optimization. Then, the global convergence of Algorithm \ref{alg:DAP-ADMM} is stated as follows:

\begin{theorem}[\text{\cite[Corollary 3.5]{syzz2025accelerating}}]\label{thm:acc-preADMM-convergence}
Suppose that Assumption \ref{asp:KKT-nonempty} holds. The sequence $\{\widebar{\bm{w}}^{k}=(\widebar{\bm{\pi}}^{k},\widebar{\bm{s}}^{k},\widebar{\bm{d}}^{k},\widebar{\bm{\theta}}^{k},\widebar{\bm{\xi}}^{k})\}$ generated by Algorithm \ref{alg:DAP-ADMM} converges to $\bm{w}^{*}=({\bm{\pi}}^{*},{\bm{s}}^{*},{\bm{d}}^{*},{\bm{\theta}}^{*},{\bm{\xi}}^{*})$ satisfying \eqref{eq:KKTsys}. Then $(\bm{\pi}^{*},\bm{s}^{*},\bm{d}^{*})$ and $(\bm{\theta}^{*},\bm{\xi}^{*})$ solve \eqref{eq:D-GOT} and \eqref{eq:D-GOT-dual}, respectively.
\end{theorem}
\begin{proof}
	See Appendix \ref{sec:App-proofs}.
\end{proof}


We next establish the convergence rate of Algorithm \ref{alg:DAP-ADMM}. To this end, we introduce the KKT \textit{residual mapping} associated with system \eqref{eq:KKTsys}, for every $\bm{w}=(\bm{\pi},\bm{s},\bm{d},\bm{\theta},\bm{\xi})\in\mathbb{W}$:
\begin{equation}\label{eq:KKTres}
\mathcal{R}(\bm{w}):=\begin{pmatrix}
\bm{\xi}-\bm{\theta} \\
\bm{s}-\mathtt{prox}_{G_{1}}(\bm{s}-\bm{\theta}) \\
\bm{d}-\mathtt{prox}_{G_{2}}(\bm{d}+\bm{\xi}) \\
\bm{\pi} - \bm{s} \\
\bm{d} - \bm{\pi}
\end{pmatrix}.
\end{equation}
Clearly, a point ${\bm{w}}^{*}$ satisfies the KKT system \eqref{eq:KKTsys} if and only if $\mathcal{R}(\bm{w}^*) = \bm{0}$. Then, we introduce a self-adjoint positive semi-definite linear operator $\mathcal{M}:\mathbb{W}\to\mathbb{W}$ by
\begin{equation}\label{eq:precond}
\mathcal{M}:= \begin{bmatrix}
\Sigma_{1} + \Sigma_{2} & O & O & -I_{|\mathcal{E}|} & I_{|\mathcal{E}|} \\
O & O & O & O & O \\
O & O & O & O & O \\
-I_{|\mathcal{E}|} & O & O & \Sigma_{1}^{-1} & O \\
I_{|\mathcal{E}|} & O & O & O & \Sigma_{2}^{-1} \\
\end{bmatrix}.
\end{equation}
Finally, to evaluate the primal objective function value error, we introduce the following gap function
\begin{equation}\label{eq:gap-fun}
\begin{aligned}[b]
h(\widebar{\bm{\pi}}^{k},\widebar{\bm{s}}^{k},\widebar{\bm{d}}^{k}) &:= G_{1}(\widebar{\bm{s}}^{k}) + G_{2}(\widebar{\bm{d}}^{k}) \\
&\quad {}- G_{1}({\bm{s}}^{*}) - G_{2}({\bm{d}}^{*}), \quad k\geq0,
\end{aligned}
\end{equation}
where $\{(\widebar{\bm{\pi}}^{k},\widebar{\bm{s}}^{k},\widebar{\bm{d}}^{k})\}$ is the sequence generated by Algorithm \ref{alg:DAP-ADMM} and $({\bm{\pi}}^{*},{\bm{s}}^{*},{\bm{d}}^{*})$ is its limit. Based on these notations, we state the convergence rate of Algorithm \ref{alg:DAP-ADMM} in terms of both KKT residual and primal objective function value gap:
\begin{theorem}\label{thm:acc-preADMM-rate}
Suppose that Assumption \ref{asp:KKT-nonempty} holds. Let $\{\widebar{\bm{w}}^{k}\}=\{(\widebar{\bm{\pi}}^{k},\widebar{\bm{s}}^{k},\widebar{\bm{d}}^{k},\widebar{\bm{\theta}}^{k},\widebar{\bm{\xi}}^{k})\}$ be the sequence generated by Algorithm \ref{alg:DAP-ADMM}, and let ${\bm{w}}^{*}=({\bm{\pi}}^{*},{\bm{s}}^{*},{\bm{d}}^{*},{\bm{\theta}}^{*},{\bm{\xi}}^{*})$ be the limit point of the sequence $\{\widebar{\bm{w}}^{k}\}$ and $R_0=\left\|\bm{w}^0-\bm{w}^*\right\|_{\mathcal{M}}$. Then for all $k \geq 0$, we have the following bounds:
\begin{equation}\label{eq:rate-of-KKTres}
\big\|\mathcal{R}(\widebar{\bm{w}}^k)\big\| \leq\left(\frac{\lambda_{\max}(\Sigma)\sqrt{2}+1}{\sqrt{\lambda_{\min}(\Sigma)}}\right) \frac{2 R_0}{\rho(k+1)},
\end{equation}
and
\begin{equation}\label{eq:rate-of-obj-gap}
\begin{aligned}[b]
h(\widebar{\bm{\pi}}^{k},\widebar{\bm{s}}^{k},\widebar{\bm{d}}^{k}) &\geq - \frac{2 R_0\|(\bm{\theta}^{*},\bm{\xi}^{*})\|}{\sqrt{\lambda_{\min}(\Sigma)}\rho(k+1)}, \\
h(\widebar{\bm{\pi}}^{k},\widebar{\bm{s}}^{k},\widebar{\bm{d}}^{k}) &\leq \frac{2 R_0\big(\sqrt{\lambda_{\max}(\Sigma)} R_0 + \|(\bm{\theta}^{*},\bm{\xi}^{*})\|\big)}{\sqrt{\lambda_{\min}(\Sigma)}\rho(k+1)},
\end{aligned}
\end{equation}
where $\Sigma:= [\Sigma_{1}, O; O, \Sigma_{2}]$, $\Sigma_{1}, \Sigma_{2}$ are given in \eqref{eq:def_Sigma}, and $\lambda_{\max}(\Sigma)$ and $\lambda_{\min}(\Sigma)$ denote the largest and smallest eigenvalues of $\Sigma$, respectively.
\end{theorem}
\begin{proof}
See Appendix \ref{sec:App-proofs}.
\end{proof}

\subsection{A Distributed Self-Adaptive Mechanism for Penalty Parameters}\label{sec:self-ada}

In this section, we describe a heuristic self-adaptive mechanism for the penalty parameters for the proposed DAP-ADMM. It is known that the efficiency of ADMM-type algorithms is highly sensitive to the penalty parameter. Therefore, one would prefer to adaptively update such a parameter using the global KKT residual to achieve better performance. However, the KKT residual mapping \eqref{eq:KKTres} involves variables $\bm{w}$ and functions $G_1,G_2$ distributed across multiple nodes, making it difficult to access in such a distributed system.

To embed this mechanism into decentralized ADMM, we use multiple local penalty parameters $\bm{\sigma}^{s}$ and $\bm{\sigma}^{d}$ and adjust them via \textit{local} KKT residuals. Recall that the proximal mappings for $G_1$ and $G_{2}$ can be rewritten as projections onto $\mathcal{U}$ and $\mathcal{V}$, respectively:
\begin{equation*}
	\begin{aligned}
		\mathtt{prox}_{G_{1}}({\bm{x}}) &= \Pi_{\mathcal{U}}\bigg(\frac{{\bm{x}} - \bm{c}^{s}}{1 + \lambda^{s}}\bigg), \\
		\mathtt{prox}_{G_{2}}({\bm{x}}) &= \Pi_{\mathcal{V}}\bigg(\frac{{\bm{x}} - \bm{c}^{d}}{1 + \lambda^{d}}\bigg).
	\end{aligned}
\end{equation*}
Decomposing \eqref{eq:KKTres} by node gives the local residual mappings for supply node $i\in\mathcal{I}$ and demand node $j\in\mathcal{J}$:
\begin{equation*}
	\begin{aligned}
	\mathcal{R}_{i}(\bm{w}) &= \begin{pmatrix}
		\dfrac{\sqrt{2}}{2}\big(\bm{\xi}_{(i,:)} - \bm{\theta}_{(i,:)}\big) \\[1mm]
		\bm{s}_{(i,:)} - \Pi_{\mathcal{U}_{i}}\bigg(\dfrac{\bm{s}_{(i,:)}-\bm{\theta}_{(i,:)} - \bm{c}_{(i,:)}^{s}}{1 + \lambda^{s}}\bigg) \\[1mm]
		\bm{\pi}_{(i,:)} - \bm{s}_{(i,:)}
		\end{pmatrix}, \\ 
	\mathcal{R}_{j}(\bm{w}) &= \begin{pmatrix}
		\dfrac{\sqrt{2}}{2}\big(\bm{\xi}_{(:,j)} - \bm{\theta}_{(:,j)}\big) \\[1mm]
		\bm{d}_{(:,j)} - \Pi_{\mathcal{V}_{j}}\bigg(\dfrac{\bm{d}_{(:,j)}+\bm{\xi}_{(:,j)} - \bm{c}_{(:,j)}^{d}}{1 + \lambda^{d}}\bigg) \\[1mm]
		\bm{d}_{(:,j)} - \bm{\pi}_{(:,j)}
		\end{pmatrix}.
	\end{aligned}
\end{equation*}
These mappings use only local variables, and
\begin{equation*}
	\big\|\mathcal{R}(\bm{w})\big\| = \sqrt{\sum\nolimits_{i\in\mathcal{I}}\big\|\mathcal{R}_{i}(\bm{w})\big\|^{2} + \sum\nolimits_{j\in\mathcal{J}}\big\|\mathcal{R}_{j}(\bm{w})\big\|^{2}}.
\end{equation*}
Following residual-balancing strategies \cite[Section 4.4]{lam2018fast} and \cite[Section 5.1]{yang2021fast}, consider supply node $i\in\mathcal{I}$. For $\bm{w}\in\mathbb{W}$, define
\begin{equation*}
	\begin{aligned}	
		\Delta_{p1,i}^{s}(\bm{w}) & = \bm{s}_{(i,:)} - \Pi_{\mathcal{U}_{i}}\left(\frac{\bm{s}_{(i,:)}-\bm{\theta}_{(i,:)} - \bm{c}_{(i,:)}^{s}}{1+\lambda^{s}}\right), \\
		\Delta_{p2,i}^{s}(\bm{w}) & = \frac{\sqrt{2}}{2}\left(\bm{\xi}_{(i,:)} - \bm{\theta}_{(i,:)}\right), \\
		\Delta_{p,i}^{s}(\bm{w}) & := \sqrt{\left\|\Delta_{p1,i}^{s}(\bm{w})\right\|^{2} + 
			\frac{1}{2}\left\|\Delta_{p2,i}^{s}(\bm{w})\right\|^{2}}, \\
		\Delta_{d,i}^{s}(\bm{w}) & := \left\|\bm{\pi}_{(i,:)} - \bm{s}_{(i,:)}\right\|.
	\end{aligned}
\end{equation*}
Initialize $\sigma_{i}^{s,0}=1$. If $\mathrm{mod}(k,\bar{k})=0$, compute $\chi_{i}^{s,k+1}=\Delta_{p,i}^{s}(\widebar{\bm{w}}^{k})/\Delta_{d,i}^{s}(\widebar{\bm{w}}^{k})$ and update as follows:
\begin{equation*}
	\sigma_{i}^{s, k+1} = \begin{cases}
		\eta^{-1}\sigma_{i}^{s, k}, & \text { if } \chi_{i}^{s, k+1}>2, \\
		\eta\sigma_{i}^{s, k}, & \text { if } 1/{\chi_{i}^{s, k+1}}>2, \\
		\sigma_{i}^{s, k}, & \text { otherwise.}
	\end{cases}
\end{equation*}
In the update above, $\eta>0$. We use $\bar{k}=50$ and $\eta=1.5$, and restrict $\sigma_{i}^{s,k+1}$ to $[10^{-3},10^{3}]$. More generally, if $0<\underline{\sigma}\leq\sigma_{i}^{s},\sigma_{j}^{d}\leq\overline{\sigma}<\infty$, then \eqref{eq:def_Sigma} gives $\lambda_{\min}(\Sigma)\geq\underline{\sigma}$ and $\lambda_{\max}(\Sigma)\leq\overline{\sigma}$, so the safeguards control the constants in Theorem \ref{thm:acc-preADMM-rate}. Demand nodes use the analogous update. Thus, each node adapts its penalty using only local residual information, without global residual aggregation.

\section{Numerical Experiments}\label{sec:num-exp}

In this section, we first evaluate the efficiency of the proposed solver-free method for the ADMM subproblem \eqref{eq:DAP-ADMM-SD}. We then conduct an ablation study to explore the acceleration mechanism within our DAP-ADMM (Algorithm \ref{alg:DAP-ADMM}). Finally, we compare DAP-ADMM with the distributed ADMM \cite{zhang2019consensus}, denoted by D-ADMM for solving problem \eqref{eq:D-GOT}. The primary experimental settings are described below.

\vspace{2mm}

\noindent\textbf{GOT Instance Generation:} Given integers $m,n > 1$ and a connectivity ratio $\iota\in(0,1]$, we randomly generate an undirected connected bipartite graph $\mathcal{G}(\mathcal{I},\mathcal{J},\mathcal{E})$ such that $|\mathcal{I}| = m$, $|\mathcal{J}|=n$ and $|\mathcal{E}| = \iota|\mathcal{I}|\cdot|\mathcal{J}|$. Next, we sample $\bm{\pi}^{\mathrm{feas}}\in\mathbb{R}^{|\mathcal{E}|}$ from the standard uniform distribution and normalize it to have an unit sum. Then, we compute $a_{i} = \sum_{j \in \mathcal{J}_{i}}\pi_{ij}^{\mathrm{feas}}$ for all $i\in\mathcal{I}$ and $b_{j} = \sum_{i \in \mathcal{I}_{j}}\pi_{ij}^{\mathrm{feas}}$ for all $j\in\mathcal{J}$, and set the capacity bounds as:
\begin{equation}\label{eq:lu-generate}
	\begin{aligned}
		& l_{i}^{s}=0.5a_{i},\quad u_{i}^{s}=1.5a_{i}, \quad \forall i\in\mathcal{I}, \\
		& l_{j}^{d}=0.5b_{j},\quad u_{j}^{d}=1.5b_{j}, \quad \forall j\in\mathcal{J}.
	\end{aligned}
\end{equation}
Following the methodology described in \cite[Section 5.1]{yang2025inexact}, we first generate a ``full-sized'' cost matrix $C^{\mathrm{full}}\in\mathbb{R}^{m\times n}$, and then extract entries with indices $(i,j)\in\mathcal{E}$ to formulate the cost vector $\bm{c} = [C_{ij}^{\mathrm{full}}]_{(i,j)\in\mathcal{E}}$ and set $\bm{c}^{s}=\bm{c}^{d} = {\bm{c}}/{2}$.

\vspace{2mm}

\noindent\textbf{Termination Criterion:} To ensure a fair comparison, all ADMM-type algorithms are terminated when it returns a point ${\bm{w}}^{k}=({\bm{\pi}}^{k},{\bm{s}}^{k},{\bm{d}}^{k},{\bm{\theta}}^{k},{\bm{\xi}}^{k})$ such that $\|\mathcal{R}({\bm{w}}^{k})\| < 10^{-6}$, where $\mathcal{R}(\cdot)$ is defined in \eqref{eq:KKTres}, or when the maximum of $10,000$ iterations is reached.

\subsection{Efficiency of the Solver-Free Method} 
The ADMM-subproblems \eqref{eq:DAP-ADMM-S} and \eqref{eq:DAP-ADMM-D} share the same structure and can be formulated as follow (see also problem \eqref{eq:simplex-like} for more details)
\begin{equation}\label{eq:simplex-like-2}
\begin{aligned}
\min\limits_{\bm{x}\in\mathbb{R}^{p}} & \quad {\mathtt f}(\bm{x}):=\frac{1}{2}\left\|\bm{x} - \widehat{\bm{x}} \right\|^{2}, \\
\mathrm{s.t.} & \quad l \leq \mathbf{1}^{\top}\bm{x} \leq u, \quad \bm{x}\geq0.
\end{aligned}
\end{equation}
To evaluate the efficiency of the proposed solver-free method in Algorithm \ref{alg:simplex-like}, we adopt the commercial solver Gurobi as a baseline. We generate $1000$ instances of projection problem \eqref{eq:simplex-like-2} with varying scales $p\in\{10000,\; 20000,\; 30000,\; 40000,\; 50000\}$. For each instance, $\widehat{\bm{x}}\in\mathbb{R}^{p}$ is sampled from the standard Gaussian distribution, while. an auxiliary vector $\bm{x}^{\rm feas}\in \mathbb{R}^{p}$ is sampled uniformly to set $l = 0.5\cdot\bm{1}^{\top}\bm{x}^{\rm feas}$ and $u = 1.5\cdot\bm{1}^{\top}\bm{x}^{\rm feas}$.

Table \ref{Table:solver-free} reports the \textit{total} computational time together with the averaged residuals of feasibility violation and normalized objective value, defined by:
\begin{equation*}\small
	\begin{aligned}
		& \mathtt{feas}(\bm{x}) = \max\big\{\left\|\min\{\bm{x}, 0\}\right\|, ... \\
		& \qquad  \qquad \qquad \qquad \|\max\{l-\bm{1}^{\top}\bm{x}, 0\}\|, \, \|\min\{u-\bm{1}^{\top}\bm{x}, 0\}\|\big\} , \\
		& \mathtt{nobj}(\bm{x}) = \big( 1 + |{\mathtt f}(\bm{x})| + |{\mathtt f}(\bm{x}^{*})| \big)^{-1}\big|{\mathtt f}(\bm{x}) - {\mathtt f}(\bm{x}^{*})\big|,
	\end{aligned}
\end{equation*}
where $\bm{x}^*$ is the solution returned by Gurobi. The results confirm the superior efficiency of our solver-free method: the two methods attain comparable numerical accuracy, whereas our solver-free method is substantially faster at every tested scale. This advantage is particularly relevant in the distributed setting, since the projection problems \eqref{eq:simplex-like-2} (or \eqref{eq:simplex-like}) are solved repeatedly at every node and every ADMM iteration. Thus, replacing a generic QP solver with the proposed solver-free method can significantly reduce local computational overhead, especially for large-scale graph-structured OT problems.

\begin{table}[ht]
	\caption{Comparison of Gurobi and the solver-free method (Algorithm \ref{alg:simplex-like}) for projection problem \eqref{eq:simplex-like-2} across scales $p$. The ``\texttt{time}'' column reports total computational time in seconds.}
	\label{Table:solver-free}
	\centering
	\renewcommand{\arraystretch}{1.05}
	\setlength{\tabcolsep}{3.3pt}
	\scriptsize
	\resizebox{0.9\columnwidth}{!}{
	\begin{tabular}{c lll lll}
		\toprule
		& \multicolumn{3}{c}{\textbf{Gurobi}}
		& \multicolumn{3}{c}{\textbf{Solver-free}} \\
		\cmidrule(lr){2-4}\cmidrule(lr){5-7}
		$p$
		& \texttt{feas} & \texttt{nobj} & \texttt{time}
		& \texttt{feas} & \texttt{nobj} & \texttt{time} \\
		\midrule
		$10{,}000$
		& 5.43e-13 & 0 & 66.36
		& 2.03e-12 & 5.54e-11 & \textbf{0.38} \\

		$20{,}000$
		& 2.66e-12 & 0 & 78.58
		& 5.35e-12 & 6.43e-11 & \textbf{0.55} \\

		$30{,}000$
		& 5.41e-12 & 0 & 91.34
		& 1.00e-11 & 7.07e-11 & \textbf{0.70} \\

		$40{,}000$
		& 1.20e-11 & 0 & 104.18
		& 1.60e-11 & 7.49e-11 & \textbf{0.91} \\

		$50{,}000$
		& 2.01e-11 & 0 & 116.74
		& 1.80e-11 & 7.71e-11 & \textbf{1.03} \\
		\bottomrule
	\end{tabular}}
\end{table}

\subsection{Ablation Study of DAP-ADMM}\label{sec:ablation}
We next conduct an ablation study to investigate how the individual acceleration mechanisms affect the practical behavior of DAP-ADMM. All ablation variants, summarized in Table \ref{Table:DAP-ablation-config}, share the basic distributed updates in \eqref{eq:DAP-ADMM-S}-\eqref{eq:DAP-ADMM-X2}, but differ in the use of Halpern iteration, over-relaxation, periodic anchor restart, and adaptive local penalties. In particular, $\rho$ controls the over-relaxation in \eqref{eq:ex-Hal-S}-\eqref{eq:ex-Hal-D}, while the adaptive variant updates the local penalty matrices $\Sigma_1$ and $\Sigma_2$ through the distributed mechanism in Section  \ref{sec:self-ada}. Hence, the ablation separates the effect of the basic acceleration operations from that of the restart and penalty-adaptation strategies.

\begin{table}[ht]
	\centering
	\caption{Algorithmic configurations in the ablation study. Here, H, OR, R, and Ada denote the Halpern iteration, over-relaxation, periodic anchor restart and adaptive local penalty parameters, respectively; $\rho$ is the over-relaxation parameter.}
	\label{Table:DAP-ablation-config}
	\setlength{\tabcolsep}{3.5pt}
	\renewcommand{\arraystretch}{1.05}
	\begin{tabular}{lcccc}
		\toprule
		\textbf{method} & \textbf{H} & $\rho$ & \textbf{anchor} & \textbf{penalty} \\
		\midrule
		D-ADMM               & no  & $1$   & --        & fixed \\
		D-ADMM $+$ H         & yes & $1$   & $\bm{w}^0$ & fixed \\
		D-ADMM $+$ OR        & no  & $1.9$ & --        & fixed \\
		DAP-ADMM (fixed)     & yes & $2$   & $\bm{w}^0$ & fixed \\
		DAP-ADMM (R $+$ Ada) & yes & $2$   & restarted & adaptive \\
		\bottomrule
	\end{tabular}
\end{table}

\begin{figure*}[ht]
	\centering
	\subfloat[$\lambda=0$]{
		\includegraphics[width=0.32\textwidth]
		{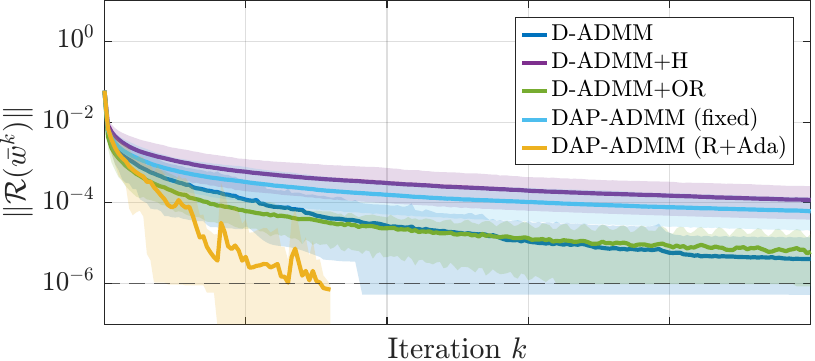}}
	\hfill
	\subfloat[$\lambda=0.1$]{
		\includegraphics[width=0.32\textwidth]
		{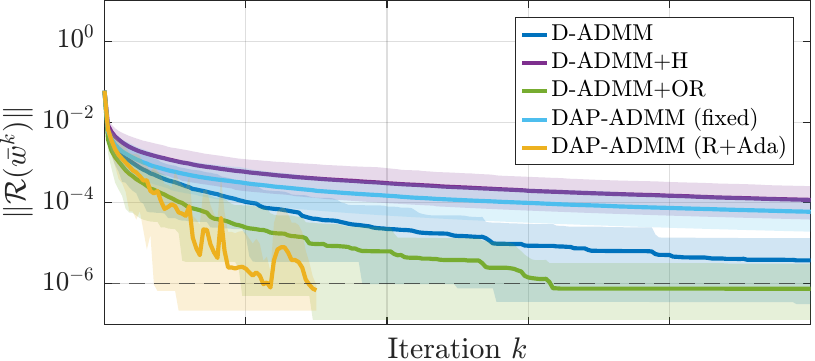}}
	\hfill
	\subfloat[$\lambda=1$]{
		\includegraphics[width=0.32\textwidth]
		{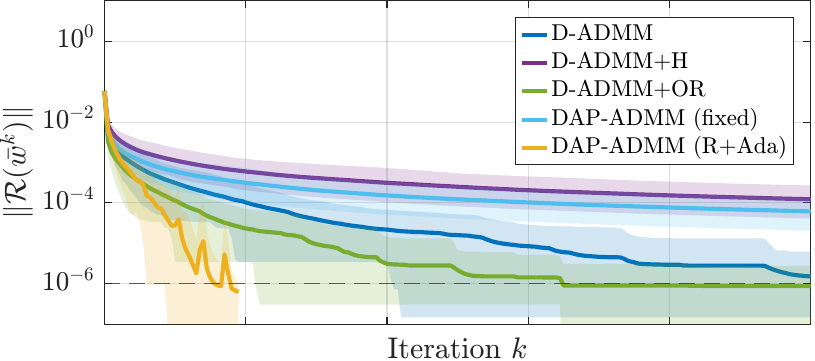}}
	\caption{KKT residual evolution for the ablation study on regular bipartite graphs with $m=n=300$ and $\iota=0.2$. Each solid curve represents the mean over $10$ paired instances, and the shaded region spans the minimum to the maximum across those instances.}
	\label{fig:DAP-ablation}
\end{figure*}
We conduct the study on regular bipartite graphs with $m=n=300$ and $\iota=0.2$. For each $\lambda\in\{0,0.1,1\}$, all variants are evaluated on $10$ randomly generated instances. Figure \ref{fig:DAP-ablation} reveals a consistent progression among the five configurations. Taking D-ADMM as the baseline, over-relaxation is able to improve the practical convergence, whereas the Halpern iteration without additional safeguards does not necessarily translate its theoretical residual acceleration into faster numerical convergence. In fact, Halpern iteration alone, or together with over-relaxation under fixed penalties, may slow the observed convergence.

This observation highlights a distinction between the theoretical acceleration mechanism and its practical realization. Once periodic anchor restart and adaptive local penalty updates are incorporated, the convergence behavior improves markedly: DAP-ADMM reaches the prescribed tolerance in all displayed cases and requires substantially fewer iterations than D-ADMM. These results therefore support the use of restart and local penalty adaptation as practical complements to the Halpern-based acceleration. Accordingly, DAP-ADMM in the following comparisons refers to this restarted and adaptive implementation.

\subsection{Comparison Between Distributed ADMMs}

In this part of the experiments, we compare the numerical performance of DAP-ADMM with the existing D-ADMM proposed in \cite{zhang2019consensus} for solving the synthetic GOT problems \eqref{eq:D-GOT}. Both algorithms adopt the solver-free method to compute their ADMM-subproblems. We check their KKT residual ``$\|\mathcal{R}(\widebar{\bm{w}}^{k})\|$'' and normalized primal objective value ``$\texttt{nobj}$'' at every $50$ iterations, where $\texttt{nobj}$ is defined by
\begin{equation*}
	\small
	\texttt{nobj}(\widebar{\bm{w}}^{k}) = \frac{\big|G_{1}(\widebar{\bm{s}}^{k}) + G_{2}(\widebar{\bm{d}}^{k}) - G_{1}({\bm{s}}^{*}) - G_{2}({\bm{d}}^{*})\big|}{1 + \big|G_{1}(\widebar{\bm{s}}^{k}) + G_{2}(\widebar{\bm{d}}^{k})\big| + \big|G_{1}({\bm{s}}^{*}) + G_{2}({\bm{d}}^{*})\big|},
\end{equation*}
with ${\bm{s}}^{*}$ and ${\bm{d}}^{*}$ being the reference solution returned by Gurobi. We fix its penalty parameter at $\sigma=1$ for D-ADMM. For DAP-ADMM, we set $\rho=2$ and initialize $\bm{\sigma}=(\bm{\sigma}^{s},\bm{\sigma}^{d})=(\bm{1}_{m},\bm{1}_{n})$. Following Section \ref{sec:ablation}, the local penalty parameters of DAP-ADMM are updated every $50$ iterations, and its anchor point of Halpern iteration is restarted at the same frequency according to \cite[Section 4]{syzz2025accelerating}.

\begin{table*}[ht]
\centering
\caption{Performance of D-ADMM and DAP-ADMM for solving \eqref{eq:D-GOT} on regular bipartite graphs, averaged over $10$ instances. Here, ``$\|\mathcal{R}(\protect\widebar{\bm{w}}^{k})\|$'' is the KKT-residual norm, ``\texttt{nobj}'' is the normalized primal objective value, and ``\#'' denotes the iteration count; computational time in seconds is reported in parentheses.}
\label{Table:vsADMM-regG}
\scalebox{0.9}{
\begin{tabular}{cllllllllll}
\toprule
& & \multicolumn{3}{c}{$\lambda=0$} & \multicolumn{3}{c}{$\lambda=0.1$} & \multicolumn{3}{c}{$\lambda=1$} \\ 
\cmidrule(l){3-5} \cmidrule(l){6-8} \cmidrule(l){9-11} 
$m=n$ & \textbf{method} & $\|\mathcal{R}(\widebar{\bm{w}}^{k})\|$ & \texttt{nobj} & \#(\texttt{time}) & $\|\mathcal{R}(\widebar{\bm{w}}^{k})\|$ & \texttt{nobj} & \#(\texttt{time}) & $\|\mathcal{R}(\widebar{\bm{w}}^{k})\|$ & \texttt{nobj} & \#(\texttt{time})  \\ 
\midrule
\multirow{2}{*}{100} & D-ADMM & 9.76e-07 & 1.86e-07 & 5505(1.42) & 6.16e-07 & 2.55e-07 & 5105(1.46) & 7.92e-07 & 3.63e-07 & 3755(0.44) \\ 
& DAP-ADMM & \textbf{3.86e-07} & \textbf{1.67e-07} & \textbf{850(0.26)} & \textbf{1.99e-07} & \textbf{2.23e-07} & \textbf{820(0.25)} & \textbf{3.70e-07} & \textbf{2.97e-07} & \textbf{735(0.09)} \\ 
\midrule
\multirow{2}{*}{300} & D-ADMM & 2.67e-06 & 7.67e-07 & 8870(11.27) & 2.47e-06 & 7.15e-07 & 8790(11.08) & 2.27e-06 & 7.72e-07 & 8070(4.57) \\ 
& DAP-ADMM & \textbf{6.20e-07} & \textbf{2.10e-07} & \textbf{1550(2.23)} & \textbf{4.28e-07} & \textbf{2.23e-07} & \textbf{1460(2.07)} & \textbf{3.78e-07} & \textbf{3.74e-07} & \textbf{1320(0.83)} \\ 
\midrule
\multirow{2}{*}{500} & D-ADMM & 5.86e-06 & 6.55e-06 & 9450(24.81) & 5.15e-06 & 5.82e-06 & 9370(12.71) & 4.28e-06 & 5.49e-06 & 8530(8.82) \\ 
& DAP-ADMM & \textbf{6.08e-07} & \textbf{2.60e-07} & \textbf{1950(5.91)} & \textbf{5.26e-07} & \textbf{2.94e-07} & \textbf{1695(2.33)} & \textbf{6.36e-07} & \textbf{5.02e-07} & \textbf{1380(1.65)} \\ 
\midrule
\multirow{2}{*}{1000} & D-ADMM & 2.07e-05 & 5.22e-05 & 9890(139.09) & 2.01e-05 & 5.24e-05 & 9905(82.27) & 1.96e-05 & 5.43e-05 & 10000(82.86) \\ 
& DAP-ADMM & \textbf{7.77e-07} & \textbf{3.58e-07} & \textbf{2970(48.96)} & \textbf{5.44e-07} & \textbf{5.60e-07} & \textbf{2540(25.36)} & \textbf{7.20e-07} & \textbf{5.56e-07} & \textbf{2245(22.39)} \\ 
\bottomrule
\end{tabular}}
\end{table*}

We first consider regular bipartite graphs with $m=n\in\{100,300,500,1000\}$ and $\iota=0.2$. For each graph size, Table \ref{Table:vsADMM-regG} reports the average performance over $10$ random instances for $\lambda\in\{0,0.1,1\}$. DAP-ADMM consistently requires substantially fewer iterations while attaining lower KKT residuals and normalized objective errors in most tested cases. Moreover, the difference becomes more pronounced as the graph size increases, indicating that the proposed acceleration is increasingly beneficial for larger distributed problems.

\begin{table*}[ht]
	\centering
	\caption{Same as Table \ref{Table:vsADMM-regG}, but for solving problems \eqref{eq:D-GOT} on bipartite graphs with random degrees.}
	\label{Table:vsADMM-ranG}
	\scalebox{0.9}{
		\begin{tabular}{cllllllllll}
			\toprule
			& & \multicolumn{3}{c}{$\lambda=0$} & \multicolumn{3}{c}{$\lambda=0.1$} & \multicolumn{3}{c}{$\lambda=1$} \\ 
			\cmidrule(l){3-5} \cmidrule(l){6-8} \cmidrule(l){9-11} 
			$m=n$ & \textbf{method} & $\|\mathcal{R}(\widebar{\bm{w}}^{k})\|$ & \texttt{nobj} & \#(\texttt{time}) & $\|\mathcal{R}(\widebar{\bm{w}}^{k})\|$ & \texttt{nobj} & \#(\texttt{time}) & $\|\mathcal{R}(\widebar{\bm{w}}^{k})\|$ & \texttt{nobj} & \#(\texttt{time})  \\ 
			\midrule
			\multirow{2}{*}{100} & D-ADMM & 8.56e-07 & 1.96e-07 & 3850(1.21) & \textbf{4.24e-07} & 2.65e-07 & 2480(0.77) & 6.64e-07 & 4.07e-07 & 2215(0.71) \\ 
			& DAP-ADMM & \textbf{2.94e-07} & \textbf{1.52e-07} & \textbf{795(0.26)} & 5.01e-07 & \textbf{2.41e-07} & \textbf{620(0.20)} & \textbf{4.39e-07} & \textbf{3.43e-07} & \textbf{570(0.19)} \\ 
			\midrule
			\multirow{2}{*}{300} & D-ADMM & 4.56e-06 & 7.29e-06 & 9735(16.02) & 4.03e-06 & 7.16e-06 & 9090(14.93) & 2.67e-06 & 4.60e-06 & 7705(12.79) \\ 
			& DAP-ADMM & \textbf{6.00e-07} & \textbf{3.96e-07} & \textbf{1845(3.25)} & \textbf{5.05e-07} & \textbf{2.47e-07} & \textbf{1480(2.63)} & \textbf{5.45e-07} & \textbf{5.79e-07} & \textbf{1210(2.19)} \\ 
			\midrule
			\multirow{2}{*}{500} & D-ADMM & 2.67e-05 & 1.15e-04 & 9180(34.42) & 2.52e-05 & 1.14e-04 & 8925(33.54) & 2.42e-05 & 1.14e-04 & 8930(33.78) \\ 
			& DAP-ADMM & \textbf{5.97e-07} & \textbf{2.88e-07} & \textbf{2280(9.20)} & \textbf{4.87e-07} & \textbf{6.39e-07} & \textbf{2045(8.31)} & \textbf{5.79e-07} & \textbf{7.10e-07} & \textbf{1620(6.61)} \\ 
			\midrule
			\multirow{2}{*}{1000} & D-ADMM & 1.17e-04 & 1.09e-03 & 10000(158.26) & 1.15e-04 & 1.09e-03 & 10000(157.95) & 1.14e-04 & 1.10e-03 & 9690(153.28) \\ 
			& DAP-ADMM & \textbf{1.14e-05} & \textbf{2.75e-05} & \textbf{5120(91.12)} & \textbf{4.77e-06} & \textbf{7.33e-06} & \textbf{4455(79.32)} & \textbf{8.31e-07} & \textbf{2.84e-06} & \textbf{3680(65.24)} \\ 
			\bottomrule
	\end{tabular}}
\end{table*}

Table \ref{Table:vsADMM-ranG} further evaluates the two algorithms on bipartite graphs with random node degrees. Although the irregular topology leads to more heterogeneous local subproblems, DAP-ADMM largely preserves its advantage in both accuracy and iteration count. Taken together, the results on regular and random-degree graphs suggest that the combined acceleration, restart, and adaptive-penalty mechanisms improve the practical efficiency of distributed ADMM without relying on a particular graph topology, thereby reducing the communication rounds required to obtain accurate transportation plans.

\section{Conclusions}\label{sec:conclude}
In this paper, we adapted AP-ADMM to solve quadratically regularized optimal transport on bipartite graphs and developed a fully distributed implementation, termed as DAP-ADMM. Beyond this distributed extension, we developed an exact solver-free method for efficiently solving the local ADMM subproblems and a distributed self-adaptive penalty mechanism driven by locally computable KKT residuals. For fixed local penalties, we established global convergence and non-ergodic $O(1/k)$ rates, while numerical results demonstrated the effectiveness and efficiency of these developments. Future work will consider asynchronous and unreliable communications and approximate projections for high-degree nodes.

\appendices
\section{Missing Proofs in Section \ref{sec:thms}}\label{sec:App-proofs}
To present the global convergence and complexity results of DAP-ADMM in Algorithm \ref{alg:DAP-ADMM}, we rewrite problem \eqref{eq:D-GOT} as:
\begin{equation}\label{eq:temp-ADMM}
	\begin{aligned}
		\hspace{-2mm}\min_{\bm{y}\in\mathbb{R}^{|\mathcal{E}|},\bm{z}\in\mathbb{R}^{2|\mathcal{E}|}} & ~~ F(\bm{y}) + G(\bm{z}) ~~ \mathrm{ s.t. } ~~ B_{1}\bm{y} + B_{2}\bm{z} = \bm{0},
	\end{aligned}
\end{equation}
where $\bm{y} = [\bm{\pi}]$, $\bm{z} = [\bm{s}; \bm{d}]$, and
\begin{equation*}
	\begin{aligned}
		& F(\bm{y}) = 0, \quad G(\bm{z}) = G_{1}(\bm{s}) + G_{2}(\bm{d}),\\ 
		& B_{1} = \begin{bmatrix}
			-I_{|\mathcal{E}|} \\
			I_{|\mathcal{E}|}
		\end{bmatrix}, \quad B_{2} = \begin{bmatrix}
			I_{|\mathcal{E}|} & O \\ O & -I_{|\mathcal{E}|}
		\end{bmatrix}.
	\end{aligned}
\end{equation*}
We also denote $\bm{x} = [\bm{\theta}; \bm{\xi}]$. Thus, vector $\bm{w}$ in Algorithm \ref{alg:acc-preADMM} can be rewritten as $\bm{w} = (\bm{y},\bm{z},\bm{x})$. Then, the KKT system \eqref{eq:KKTsys} can be further rewritten as
\begin{equation}\label{eq:KKTsys2}
	\begin{aligned}
		-B_{1}^{\top} \bm{x} & \in \partial F\left(\bm{y}\right), \\
		-B_{2}^{\top} \bm{x} & \in \partial G\left(\bm{z}\right), \\
		\bm{0} & = - B_{1} \bm{y} - B_2 \bm{z},
	\end{aligned}
\end{equation}
where $\partial F$ and $\partial G$ are the subdifferential mappings of $F$ and $G$, respectively. Accordingly, the KKT mapping $\mathcal{T}$ defined in \eqref{eq:KKTmap-GOT} can be reformulated as
\begin{equation}\label{eq:KKTmap-GOT2}
	\mathcal{T}\bm{w}=\begin{pmatrix}
		\partial F(\bm{y})+B_{1}^{\top} \bm{x} \\
		\partial G(\bm{z})+B_{2}^{\top} \bm{x} \\
		-B_1 \bm{y}-B_2 \bm{z}
	\end{pmatrix}, \quad \forall \bm{w}=(\bm{y}, \bm{z}, \bm{x}) \in \mathbb{W} ,
\end{equation}
and the positive semi-definite linear operator $\mathcal{M}$ defined in \eqref{eq:precond} can be reformulated as
\begin{equation}\label{eq:precond2}
	\mathcal{M} := \begin{bmatrix}
		B_{1}^{\top}\Sigma B_{1} & O & B_{1}^{\top} \\
		O & O & O \\
		B_{1} & O & \Sigma^{-1}
	\end{bmatrix},
\end{equation}
where $\Sigma:= [\Sigma_{1}, O; O, \Sigma_{2}]$.

\subsection{Proof of Theorem \ref{thm:acc-preADMM-convergence}}\label{sec:App-proof-convergence}

We first recall the following result.
\begin{lemma}[\text{\cite[Lemma 3.1]{syzz2025accelerating}}]\label{lem:Lip-cons}
	Consider the operators $\mathcal{T}$ defined in \eqref{eq:KKTmap-GOT2} and $\mathcal{M}$ defined in \eqref{eq:precond2}. Then, $(\mathcal{M}+\mathcal{T})^{-1}$ is single-valued and Lipschitz continuous.
\end{lemma}

\begin{proof}
	The proof follows directly from that of \cite[Lemma~3.1]{syzz2025accelerating} by replacing $\sigma B_{2}^{*}$, $\sigma B_{1}^{*}B_{1}$, and $\sigma^{-1}$ with $\Sigma B_{2}^{\top}$, $B_{1}^{\top}\Sigma B_{1}$, and $\Sigma^{-1}$, respectively, and by setting $\Sigma_{f_{1}}$, $\Sigma_{f_{2}}$, $\mathcal{T}_{1}$, and $\mathcal{T}_{2}$ to zero. Noting that $B_{1}^{\top}\Sigma B_1\succ0$ and $B_{2}^{\top}\Sigma B_2\succ0$, the same argument yields the desired result.
\end{proof}
Based on Lemma \ref{lem:Lip-cons}, we now proceed to prove Theorem \ref{thm:acc-preADMM-convergence}, which follows directly from the proof of \cite[Proposition 3.2]{syzz2025accelerating} and \cite[Theorem 2.7]{syzz2025accelerating}.
\begin{proof}
	According to Lemma \ref{lem:Lip-cons} and \cite[Proposition 3.2]{syzz2025accelerating}, we can obtain that, under Assumption \ref{asp:KKT-nonempty}, the sequence $\big\{\widebar{\bm{w}}^k\big\}$ generated by Algorithm \ref{alg:DAP-ADMM} coincides with the sequence $\big\{\bar{\bm{w}}^k\big\}$ generated by an accelerated degenerate proximal point algorithm (acc-dPPA) \cite[Algorithm 2.1]{syzz2025accelerating} with $\alpha = 2$:
	\begin{numcases}{\text{(acc-dPPA) ~}}
		\bar{\bm{w}}^k=(\mathcal{M}+\mathcal{T})^{-1} \mathcal{M} \bm{w}^k, \label{eq:dPPA}\\
		\widehat{\bm{w}}^{k+1} = \bm{w}^k + \rho(\bar{\bm{w}}^k - \bm{w}^{k}), \label{eq:relaxation}\\
		\bm{w}^{k+1} = \frac{1}{k+2}\bm{w}^{0} + \frac{k+1}{k+2}\widehat{\bm{w}}^{k+1}. \label{eq:Halpern-iteration}
	\end{numcases}
	Thus, the global convergence of Algorithm \ref{alg:DAP-ADMM} follows from \cite[Theorem 2.7]{syzz2025accelerating}.
\end{proof}

\subsection{Proof of Theorem \ref{thm:acc-preADMM-rate}}\label{sec:App-proof-rate}
To establish the convergence rate of the primal objective value, we shall introduce the following key lemma:
\begin{lemma}[\text{\cite[Lemma 3.6]{syzz2025accelerating}}]\label{lem:obj-gap}
Suppose that Assumption \ref{asp:KKT-nonempty} holds. Let $\{\bar{\bm{w}}^{k}\} = \{(\bar{\bm{y}}^{k},\bar{\bm{z}}^{k},\bar{\bm{x}}^{k})\}$ be the sequence generated by scheme \eqref{eq:dPPA}-\eqref{eq:Halpern-iteration} and let $(\bm{y}^{*},\bm{z}^{*},\bm{x}^{*})$ be the limit point of the sequence $\{(\bar{\bm{y}}^{k},\bar{\bm{z}}^{k},\bar{\bm{x}}^{k})\}$. Then, for all $k\geq0$, we have the following bounds:
\begin{equation}\label{eq:obj-gap-ub}
    h(\bar{\bm{y}}^{k},\bar{\bm{z}}^{k}) \leq \big\langle \Sigma B_{1}(\bm{y}^{*}-\bar{\bm{y}}^{k}) - \bar{\bm{x}}^{k}, B_{1}\bar{\bm{y}}^{k} + B_{2}\bar{\bm{z}}^{k}\big\rangle,
\end{equation}
and
\begin{equation}\label{eq:obj-gap-lb}
    h(\bar{\bm{y}}^{k},\bar{\bm{z}}^{k}) \geq - \big\langle \bm{x}^{*}, B_{1}\bar{\bm{y}}^{k} + B_{2}\bar{\bm{z}}^{k}\big\rangle.
\end{equation}
\end{lemma}
\begin{proof}
	The proof also follows directly from that of \cite[Lemma 3.6]{syzz2025accelerating} by replacing $\sigma B_{1}$ and $\sigma B_{2}$ with $\Sigma B_{1}$ and $\Sigma B_{2}$, and by setting $\mathcal{T}_{1}$ and $\mathcal{T}_{2}$ to be zero operators. With these adjustments, we can apply the proof of \cite[Lemma 3.6]{syzz2025accelerating} to obtain the desired results.
\end{proof}
We now proceed with the proof of Theorem \ref{thm:acc-preADMM-rate}.
\begin{proof}
	By replacing $\sigma B_{1}$ and $\sigma B_{2}$ with $\Sigma B_{1}$ and $\Sigma B_{2}$, and letting $\mathcal{T}_{1}$ and $\mathcal{T}_{2}$ be zero operators, we can directly follow the proof of \cite[Theorem 3.7]{syzz2025accelerating} to obtain that for any $k \geq 0$,
	\begin{equation}\label{eq:rate-of-KKTres-cmpt}
		\begin{aligned}
			\big\|\mathcal{R}(\bar{\bm{w}}^k)\big\|  \leq \left(\frac{\lambda_{\max}(\Sigma)\|B_{1}\|+1}{\sqrt{\lambda_{\min}(\Sigma)}}\right) \frac{2 R_0}{\rho(k+1)}.
		\end{aligned}
	\end{equation}
	Since $\|B_{1}\|=\sqrt{2}$, this yields the inequality \eqref{eq:rate-of-KKTres}. Next, according to Lemma \ref{lem:obj-gap} and that $\mathcal{T}_{1} = \mathcal{T}_{2} = 0$, we have, for all $k \geq 0$,
	\begin{equation}\label{eq:rate-of-obj-gap-cmpt-lhs}
		\begin{aligned}
			h(\bar{\bm{y}}^k, \bar{\bm{z}}^k) & \geq -\|\bm{x}^{*}\|\big\|B_1 \bar{\bm{y}}^k + B_2 \bar{\bm{z}}^k \big\| \\
			& \geq - \frac{2 R_0\|\bm{x}^*\|}{\sqrt{\lambda_{\min}(\Sigma)} \rho(k+1)},
		\end{aligned}
	\end{equation}
	and
	\begin{equation}\label{eq:rate-of-obj-gap-cmpt-rhs}
		\begin{aligned}
			h(\bar{\bm{y}}^k, \bar{\bm{z}}^k) \leq \; \frac{2 R_0 \big(\sqrt{\lambda_{\max}(\Sigma)} R_0 + \|\bm{x}^*\|\big)}{\sqrt{\lambda_{\min}(\Sigma)} \rho(k+1)}.
		\end{aligned}
	\end{equation}
	By substituting $\bar{\bm{y}}^{k} = \widebar{\bm{\pi}}^{k}$, $\bar{\bm{z}}^{k} = (\widebar{\bm{s}}^{k}, \widebar{\bm{d}}^{k})$ and $\bm{x}^{*} = (\bm{\theta}^{*},\bm{\xi}^{*})$ into \eqref{eq:rate-of-obj-gap-cmpt-lhs} and \eqref{eq:rate-of-obj-gap-cmpt-rhs}, we obtain both sides of inequality \eqref{eq:rate-of-obj-gap}. Therefore, the proof of Theorem \ref{thm:acc-preADMM-rate} is completed.
\end{proof}


\bibliographystyle{IEEEtran}
\bibliography{Ref_GOT}


%

%
%
%

\vfill

\end{document}